\RequirePackage{fix-cm}
\documentclass[smallextended]{svjour3}       
\smartqed  
\usepackage{tikz}
\usepackage{graphicx}
\usepackage{eepic}
\usepackage{color}
\usepackage{amssymb}
\usepackage{amsmath}
\usepackage{float}
\usepackage{booktabs}
\usepackage{subfigure}
\usepackage{hyperref}
\usepackage{bm}
\usepackage{multirow}
\usepackage{threeparttable}
\usepackage{mathptmx}      
\begin{document}

\title{Nonexistence of Generalized Bent Functions From $Z_{2}^{n}$ to $Z_{m}$\thanks{K. Feng is supported by NSFC No. 11471178 and the National Lab. on Information Science and Technology of Tsinghua University.}
}


\author{Haiying Liu \and Keqin Feng \and Rongquan Feng
}


\institute{H. Liu \at
              School of Mathematical Sciences, Peking University, Beijing, 100871, China \\
              \email{hyliumath@pku.edu.cn}             \\
           \and
           K. Feng \at
              Department of Mathematical Sciences, Tsinghua University, Beijing, 100084, China \\
               \email{kfeng@math.tsinghua.edu.cn}             \\
                 \and
           R. Feng \at
             School of Mathematical Sciences, Peking University, Beijing, 100871, China \\
               \email{fengrq@math.pku.edu.cn}             \\
}


\maketitle

\begin{abstract}
Several nonexistence results on generalized bent functions $f:Z_{2}^{n} \rightarrow Z_{m}$ presented by using some knowledge on cyclotomic number fields and their imaginary quadratic subfields.
\end{abstract}

\section{Introduction}
\label{intro}
Let $m$ and $n$ be positive integers, $m\geq 2$, $Z_{m}=\mathbb{Z}/m\mathbb{Z}=\{0, 1,\ldots, m-1\}$. A mapping
 $$f(x)=f(x_{1},\ldots, x_{n}):Z_{2}^{n} \rightarrow Z_{2}\quad (x=(x_{1},\ldots, x_{n})\in Z_{2}^{n})$$
 is called a boolean function with $n$ variables. The Fourier transformation of $(-1)^{f}:Z_{2}^{n} \rightarrow \{\pm 1\}\subseteq \mathbb{C}$ is the function
 $W_{f}:Z_{2}^{n} \rightarrow \mathbb{Z}$ defined by
 $$W_{f}(y)=\sum_{x\in Z_{2}^{n}}(-1)^{f(x)+x\cdot y}\quad (y=(y_{1},\ldots, y_{n})\in Z_{2}^{n})$$
 where $x\cdot y=\sum_{i=1}^{n}x_{i}y_{i}\in Z_{2}$. $f$ is called a bent function if $\big|W_f(y)\big|=2^{\frac{n}{2}}$ holds for every $y\in Z_{2}^{n}$.

The notation of bent functions was presented by Rothaus \cite{R} in 1976, and generalized by Kumar et al. \cite{KSW} in 1985 for $f:Z_{m}^{n} \rightarrow Z_{m}$. Being a family of functions with the best nonlinearity, bent functions draw much attention and are widely investigated. They are applied in many fields such as communication theory, cryptography and closely related to coding theory and combinatorial design theory. In this paper we deal with the function $$f(x)=f(x_{1},\ldots, x_{n}):Z_{2}^{n} \rightarrow Z_{m}.$$ We call $\{m,n\}$ the type of such functions. For $m=2$, this is the usual boolean function. Let $\zeta_{m}=e^{\frac{2\pi \sqrt{-1}}{m}}\in \mathbb{C}$, we define the Fourier transformation of $\zeta_{m}^{f}:Z_{2}^{n} \rightarrow \langle\zeta_{m}\rangle\subseteq \mathbb{C}$ by $W_{f}:Z_{2}^{n} \rightarrow \mathbb{Z}[\zeta_{m}]\subseteq\mathbb{C}$,
\begin{equation}
W_{f}(y)=\sum_{x\in Z_{2}^{n}}(-1)^{x\cdot y}\zeta_{m}^{f(x)}\ \ \ \ (y\in Z_{2}^{n} ).
\label{Wdefine}
\end{equation}
By Fourier inverse transformation we have
\begin{equation}
\zeta_{m}^{f(x)}=\frac{1}{2^{n}}\sum_{y\in Z_{2}^{n}}(-1)^{x\cdot y}W_{f}(y)\ \ \ \ (x\in Z_{2}^{n} ),
\label{WIdefine}
\end{equation}
and from (\ref{Wdefine}) we get
$$\frac{1}{2^{n}}\sum_{y\in Z_{2}^{n}}\big|W_{f}(y)\big|^{2}=2^{n}.$$
Therefore the maximum value $M=max\{\big|W_{f}(y)\big|:y\in Z_{2}^{n}\}$ has lower bound $2^{\frac{n}{2}}$ and $M=2^{\frac{n}{2}}$ if and only if $\big|W_{f}(y)\big|=2^{\frac{n}{2}}$ for all $y\in Z_{2}^{n}$.

\begin{definition}\label{DefineBent}
A function $f:Z_{2}^{n} \rightarrow Z_{m}$ is called $\{m,n\}$-type generalized bent function (GBF) if $\big|W_{f}(y)\big|=2^{\frac{n}{2}}$ for all $y\in Z_{2}^{n}$.
\end{definition}

For $m=2$, this is the usual boolean bent function in \cite{R}. It is well-known that there exists bent function $f:Z_{2}^{n} \rightarrow Z_{2}$
if and only if $n$ is even. And for even number $n$, many constructions on boolean bent functions with $n$ variables have been found. One of simple examples is
$f(x)=x_{1}x_{2}+x_{3}x_{4}+\cdots+x_{2t-1}x_{2t} (n=2t)$. For $m=4$ and $m=8$, several constructions and existence results on GBFs have been found in \cite{LTH,S,KUS,ST,SM,SMGS,T}.

Our main aim in this paper is to present several nonexistence results on $\{m,n\}$-type GBF. Firstly we collect existence results in Sect. \ref{sec:2} by using previous facts given in \cite{ST,SM} and some simple observations. Then we show nonexistence results in last three sections. In Sect. \ref{sec:3} we use a result in \cite{LL} to prove that for any odd prime power $p^{l}$, there is no $\{p^{l},n\}$-type GBF for any $n\geq 1$. In Sect. \ref{sec:4} and 5 we use some basic knowledge on cyclotomic number field $K=\mathbb{Q}(\zeta_{m})$ since, for $\{m,n\}-$type function $f$, the values of $W_{f}(y)$ belong to the ring $\mathcal{O}_{K}=\mathbb{Z}[\zeta_{m}]$ of (algebraic) integers of $K$. We introduce these basic facts at the beginning of Sect. \ref{sec:4}. Then we get further nonexistence results on GBF for semiprimitive case in Sect. \ref{sec:4} and by field-descent method in Sect. 5. The last Sect. 6 is conclusion.

\section{Existence Results}
\label{sec:2}
In this section we firstly collect previous existence results on $\{m,n\}$-type GBF $f:Z_{2}^{n} \rightarrow Z_{m}\qquad(n\geq1, m\geq2)$.

\begin{lemma}\label{Lemmaexist}
(1)(\cite{R}) There exists a $\{2,n\}$-type (boolean) bent function if and only if $n$ is even.

(2)(\cite{ST}) If both $n$ and $m$ are even, then there exists a $\{m,n\}$-type GBF.

(3)(\cite{SMGS}) If there exists a $\{2,n+1\}$-type (boolean) bent function, then there exists a $\{4,n\}$-type GBF.

\end{lemma}

\begin{remark}
(A) For proof of (2), let $m=2l,n=2t$, $\sigma$ is any permutation on $Z_{2}^{l}$ and $g:Z_{2}^{l} \rightarrow Z_{m}$ is any function. Consider the following function $f=f(x,y):Z_{2}^{t}\oplus Z_{2}^{t}=Z_{2}^{n} \rightarrow Z_{m},$ $$f(x,y)=g(y)+lx\cdot\sigma(y)\ \ \ \ (x,y\in Z_{2}^{t}).$$ It is easy to see that $f$ is a GBF.

(B) The conclusion (3) has been proved in \cite{SMGS} as a consequence of more general results. Here we give a direct simple proof.

Suppose that $f=f(x,x'):Z_{2}^{n}\oplus Z_{2}=Z_{2}^{n+1}\longrightarrow Z_{2} (x\in Z_{2}^{n},x'\in Z_{2})$ is a boolean bent function.
Then for $y\in Z_{2}^{n}$ and $y'\in Z_{2}$,
\begin{equation*}
\pm 2^{\frac{n+1}{2}}=W_f(y,y')=\sum_{x\in Z_{2}^{n},x'\in Z_{2}}(-1)^{f(x,x')+x\cdot y+x'y'}
\end{equation*}
\begin{equation}\label{equ:ex}
=\sum_{x\in Z_{2}^{n}}\left((-1)^{f(x,0)+x\cdot y}+(-1)^{f(x,1)+x\cdot y+y'}\right).
\end{equation}
For $a,b\in Z_{2}=\{0,1\}$ and $y\in Z_{2}^{n}$, let
$$N_{y}(a,b)=\mathop{\Large\sum}\limits_{\tiny{\begin{array}{c}x\in Z_{2}^{n}\\
f(x,0)=a,f(x,1)=b\end{array}}}(-1)^{x\cdot y},$$
then by (\ref{equ:ex}) we get
\begin{align*}
2^{\frac{n+1}{2}}\varepsilon=&\left(N_{y}(0,0)+N_{y}(0,1)-N_{y}(1,0)-N_{y}(1,1)\right)\\
&+(-1)^{y'}\left(N_{y}(0,0)+N_{y}(1,0)-N_{y}(0,1)-N_{y}(1,1)\right)
\end{align*}
where $\varepsilon\in \{1,-1\}$. Taking $y'=0$ and $1$ we get
\begin{equation}
2\left(N_{y}(0,0)-N_{y}(1,1)\right)=2\left(N_{y}(0,1)-N_{y}(1,0)\right)=2^{\frac{n+1}{2}}\varepsilon. \label{equ:ex2}
\end{equation}
Now we consider the following function $F:Z_{2}^{n} \longrightarrow Z_{4}=\{0,1,2,3\}$ defined by
$$F(x)=\left\{\begin{array}{ll}
0,&\mbox{if~$f(x,0)=f(x,1)=0$}\\
1,&\mbox{if~$f(x,0)=0,f(x,1)=1$}\\
2,&\mbox{if~$f(x,0)=f(x,1)=1$}\\
3,&\mbox{if~$f(x,0)=1,f(x,1)=0$},
\end{array}\right.$$
then for each $y\in Z_{2}^{n}$,
\begin{align*}
W_F(y)&=\sum_{x\in Z_{2}^{n}}(-1)^{x\cdot y}\zeta_{4}^{F(x)}\ \ \ \ (\zeta_{4}=\sqrt{-1})\\
&=N_{y}(0,0)+\zeta_{4}N_{y}(0,1)-N_{y}(1,1)-\zeta_{4}N_{y}(1,0)\\
&=\left(N_{y}(0,0)-N_{y}(1,1)\right)+\zeta_{4}\left(N_{y}(0,1)-N_{y}(1,0)\right).
\end{align*}
By (\ref{equ:ex2}), we get
\begin{align*}
|W_{F}(y)|^{2}&=\left(N_{y}(0,0)-N_{y}(1,1)\right)^{2}+\left(N_{y}(0,1)-N_{y}(1,0)\right)^{2}\\
&=2^{n-1}+2^{n-1}=2^n.
\end{align*}
Therefore $F$ is a GBF.

(C) For even $n\geq2$, there exists a $\{4,n\}$-type GBF by Lemma \ref{Lemmaexist}(2). For odd $n\geq1$, there exists a bent function $f:Z_{2}^{n+1}\longrightarrow Z_{2}$ by Lemma \ref{Lemmaexist}(1), then there exists a $\{4,n\}$-type GBF by Lemma \ref{Lemmaexist}(3). Therefore we have a $\{4,n\}$-type GBF for any positive integer $n$.
\end{remark}

Next we show the following simple observation.

\begin{lemma}\label{Lemmaobservation}
(1) If there exist $\{m,n\}$-type and $\{m,n'\}$-type GBFs, then there exists a $\{m,n+n'\}$-type GBF.

(2) If there exists a $\{m,n\}$-type GBF, then for any positive integer $l$ there exists a $\{lm,n\}$-type GBF. In other words, if there is no $\{m,n\}$-type GBF, then there is no $\{m',n\}$-type GBF for any factor $m'\geq2$ of $m$.
\end{lemma}

\begin{proof}
(1) Let $f:Z_{2}^{n} \rightarrow Z_{m}$ and $f':Z_{2}^{n'} \rightarrow Z_{m}$ be GBFs. It is easy to see that
$$F:Z_{2}^{n+n'}\rightarrow Z_{m}, F(x,x')=f(x)+f'(x')\qquad(x\in Z_{2}^{n}, x'\in Z_{2}^{n'})$$ is a GBF.

(2)Given a GBF $f:Z_{2}^{n} \rightarrow Z_{m}$, consider the following well-defined function
$$F:Z_{2}^{n}\rightarrow Z_{lm}, F(x)=lf(x)\qquad(x\in Z_{2}^{n}).$$
For each $y\in Z_{2}^{n}$,
$$W_F(y)=\sum_{x\in Z_{2}^{n}}(-1)^{x\cdot y}\zeta_{lm}^{F(x)}
=\sum_{x\in Z_{2}^{n}}(-1)^{x\cdot y}\zeta_{m}^{f(x)}=W_f(y).$$
Therefore, $F$ is a GBF if $f$ is a GBF.
\end{proof}

From Lemma \ref{Lemmaexist}(1), remark(C) and Lemma \ref{Lemmaobservation} we know that if both $n$ and $m$ are even or $4\mid m$, then there exists a GBF with type $\{m,n\}$. The remaining cases are

Case $(I)$\qquad $2\nmid m$, or

Case $(II)$\qquad $2 \nmid n$ and $m \equiv 2\pmod 4$.

So far there is no GBF being constructed in the above cases by our knowledge. In the following three sections we shall show several nonexistence results for these cases.

\section{Nonexistence Results (I)}\label{sec:3}
 In this section we present a nonexistence result on GBF by using the following lemma given by Lam and Leung \cite{LL}.

\begin{lemma}\cite{LL}\label{main theorem}
Let $m=p_{1}^{a_{1}}\cdots p_{s}^{a_{s}}$, where $p_{1},\ldots,p_{s}$ are distinct prime numbers and $a_{i}\geq1 (1\leq i\leq s)$. Let $S=\{\zeta_{m}^{i}:0\leq i\leq m-1\}$. Then for any  integer $n\geq1$, the equation
$x_{1}+\cdots+x_{n}=0$
has a solution $x_{i}\in S (1\leq i\leq n)$ if and only if $n=n_{1}p_{1}+\cdots+n_{s}p_{s}$ has a non-negative integer solution $(n_{1},\ldots,n_{s})$.
\end{lemma}

\begin{theorem}\label{usemain}
Let $m=p_{1}^{a_{1}}\cdots p_{s}^{a_{s}}$, where $p_{1},\ldots,p_{s}$ are distinct odd prime numbers and $a_{i}\geq1 (1\leq i\leq s)$.
If
$$2^{n}=n_{1}p_{1}+\cdots+n_{s}p_{s}$$
has no non-negative integer solution $(n_{1},\ldots,n_{s})$, then there is no GBF with type $\{m,n\}$.
\end{theorem}

\begin{proof}
Suppose that $f:Z_{2}^{n} \rightarrow Z_{m}$ is a GBF. Then $\big|W_{f}(y)\big|^{2}=2^{n}$ for any $y\in Z_{2}^{n}$. By Fourier inverse transformation (\ref{WIdefine}), for $0\neq a \in Z_{2}^{n}$, we have
 \begin{align*}
 \sum_{x\in Z_{2}^{n}}\zeta_{m}^{f(x+a)-f(x)}&=\frac{1}{2^{2n}}\sum_{x\in Z_{2}^{n}}\sum_{y\in Z_{2}^{n}}W_{f}(y)(-1)^{(x+a)\cdot y}\sum_{z\in Z_{2}^{n}}\overline{W_{f}(z)}(-1)^{x\cdot z}\\
&=\frac{1}{2^{2n}}\sum_{y,z\in Z_{2}^{n}}W_{f}(y)\overline{W_{f}(z)}(-1)^{a\cdot y}\sum_{x\in Z_{2}^{n}}(-1)^{x\cdot (y+z)}\\
&=\frac{1}{2^{n}}\sum_{y\in Z_{2}^{n}}W_{f}(y)\overline{W_{f}(y)}(-1)^{a\cdot y}=\frac{1}{2^{n}}\sum_{y\in Z_{2}^{n}}\big|W_{f}(y)\big|^{2}(-1)^{a\cdot y}\\
&=\sum_{y\in Z_{2}^{n}}(-1)^{a\cdot y}=0\qquad (\mbox{since\ } 0\neq a \in Z_{2}^{n}).
\end{align*}
The left-hand side is a sum of $2^{n}$ elements and each element is a power of $\zeta_{m}$. By Lemma \ref{main theorem}, the equation $2^{n}=n_{1}p_{1}+\cdots+n_{s}p_{s}$ has a non-negative integer solution $(n_{1},\ldots,n_{s})$. This completes the proof of Theorem \ref{usemain}.
\end{proof}

\begin{corollary}\label{maincor}
There is no GBF with type $\{m,n\}$ provided one of the following conditions on $n$ and $m$  is satisfied.

(1) $m$ is an odd prime power and $n$ is any positive integer;

(2) $m=p_{1}^{a_{1}}\cdots p_{s}^{a_{s}}$, $3\leq p_{1}<p_{2}<\cdots< p_{s}$ and $2^{n}<p_{1}+p_{2}$;

(3) $m=p_{1}^{a_{1}}p_{2}^{a_{2}}, p_{1}\neq p_{2}, a_{1},a_{2}\geq1, 2^{n}=p_{1}p_{2}-m_{1}p_{1}-m_{2}p_{2}$ for some positive integers $m_{1}$ and $m_{2}$.
\end{corollary}
\begin{proof}
This is a direct consequence of Theorem \ref{usemain} since

for (1), $2^{n}=px (2\nmid p\geq3)$ has no solution $x\in \mathbb{Z}$;

for (2), if $2^{n}=n_{1}p_{1}+\cdots+n_{s}p_{s} (n_{i}\geq 0)$, then $n_{i}\geq 1$ for at least two $i$. Therefore $2^{n}\geq p_{1}+p_{2}$;

for (3), if $p_{1}p_{2}-m_{1}p_{1}-m_{2}p_{2}=2^{n}=n_{1}p_{1}+n_{2}p_{2}$, then $p_{1}p_{2}-p_{1}-p_{2}=p_{1}(n_{1}+m_{1}-1)+p_{2}(n_{2}+m_{2}-1)$ which contradicts to the well-known fact that $p_{1}p_{2}-p_{1}-p_{2}$ is the largest integer $N$ such that $N=p_{1}x_{1}+p_{2}x_{2}$ has no non-negative integer solution $(x_{1},x_{2})$. Therefore, $2^{n}=n_{1}p_{1}+n_{2}p_{2}$ has no non-negative integer solution $(n_{1},n_{2})$.
\end{proof}

\begin{example}
By Corollary \ref{maincor} it can be computed that there is no $\{7^{a_{1}}13^{a_{2}},n\}$- type GBF for all $a_{1},a_{2}\geq1$ and $1\leq n\leq 6$.
\end{example}

\section{Nonexistence Results (II): Semiprimitive Case}\label{sec:4}
In this and next sections we consider the case $2 \nmid n$ and $2\nmid m$ or $m \equiv 2\pmod 4$. Firstly we introduce some basic facts on cyclotomic number field $\mathbb{Q}(\zeta_{m})$ where $\mathbb{Q}$ is the field of rational numbers. We refer to books \cite{IR} and \cite{W} for details. If $m=2t, 2\nmid t$, we have $\zeta_{t}=\zeta_{m}^{2}$ and $-\zeta_{m}=\zeta_{m}^{\frac{m}{2}+1}=\zeta_{2t}^{t+1}=\zeta_{t}^{\frac{t+1}{2}}$. Therefore,  $\mathbb{Q}(\zeta_{m})=\mathbb{Q}(\zeta_{t})$ and we can assume $K=\mathbb{Q}(\zeta_{m})$ where $m$ is odd.

\vspace{2 ex}
FACT(1). For the cyclotomic field $K=\mathbb{Q}(\zeta_{m})$ where $m\geq 3$ and $m\not\equiv 2\pmod 4$, $K/\mathbb{Q}$ is a Galois extension with degree $[K:\mathbb{Q}]=\varphi(m)$ where $\varphi(m)$ is the Euler (totient) function. The Galois group of
$K/\mathbb{Q}$ is
$$Gal(K/\mathbb{Q})=\{\sigma_{a}:a\in Z_{m}^{\ast}\}$$ where $Z_{m}^{\ast}=\{1\leq a \leq m-1: (a,m)=1\}$ is the multiplicative group of units in the ring $Z_{m}, |Z_{m}^{\ast}|=\varphi(m),$ and the automorphism $\sigma_{a}$ of $K$ is defined by $\sigma_{a}(\zeta_{m})=\zeta_{m}^{a}$. We have $\sigma_{b}=\sigma_{ab}$ for $a,b\in Z_{m}^{\ast}$. Therefore we have the following isomorphism of groups
$$Z_{m}^{\ast}\cong Gal(K/\mathbb{Q}), \qquad a\mapsto \sigma_{a}.$$

\vspace{2 ex}
FACT(2). Let $L$ be any (algebraic) number field which means that $L/\mathbb{Q}$ is a finite extension. An element $\alpha \in L$ is called an (algebraic) integer in $L$ if there exists a monic polynomial $f(x)\in \mathbb{Z}[x]$ such that $f(\alpha)=0$. The set of all integers in $L$ is a (commutative) ring, called the ring of integers in $L$ and denoted by $\mathcal{O}_{L}$. For $K=\mathbb{Q}(\zeta_{m})$ we have $$\mathcal{O}_{K}=\mathbb{Z}[\zeta_{m}]=\mathbb{Z}\oplus\zeta_{m}\mathbb{Z}\oplus\cdots \oplus\zeta_{m}^{\varphi(m)-1}\mathbb{Z}.$$
Namely, each integer $\alpha\in \mathcal{O}_{K}$ can be uniquely expressed as $$\alpha=a_0+a_1\zeta_{m}+\cdots+a_{\varphi(m)-1}\zeta_{m}^{\varphi(m)-1}\quad (a_{i} \in \mathbb{Z}).$$
We denote the group of units in $\mathcal{O}_{L}$ by $U_{L}$, and the group of roots of 1 in $\mathcal{O}_{L}$ by $W_{L}$. Then $W_{L}\subseteq U_{L}$, and for $K=\mathbb{Q}(\zeta_{m}) (2\nmid m)$ $W_{K}=\langle \zeta_{2m}\rangle = \{\zeta_{2m}^{i}: 0\leq i\leq 2m-1\}$ and for $\alpha\in\mathcal{O}_{K}$, we have $\alpha\in W_{K}$ if and only if $|\sigma(\alpha)|=1$ for each $\sigma\in Gal(K/\mathbb{Q})$.

\vspace{2 ex}
FACT(3). Let $L$ be a number field. Any idea $A\neq (0)$ of $\mathcal{O}_{L}$ can be expressed as $A=P_{1}P_{2}\cdots P_{s}$ where $P_{1},P_{2},\ldots, P_{s}$ are (non-zero) prime ideals of $\mathcal{O}_{L}$. The expression is unique up to the order of $P_{i} (1\leq i \leq s)$. Therefore $A$ can be uniquely (up to the order) expressed as $$A=P_{1}^{a_{1}}\cdots P_{g}^{a_{g}}$$ where $P_{1},\ldots, P_{s}$ are distinct (non-zero) prime ideals of $\mathcal{O}_{L}$ and $a_{i}\geq1 (1\leq i \leq g)$. In other words, the set $S(L)$ of all non-zero ideals of $\mathcal{O}_{L}$ is a free multiplicative commutative semigroup with a basis $B(L)$, the set of all non-zero prime ideals of $\mathcal{O}_{L}$. Such semigroup $S(L)$ can be extended to the (commutative) group $I(L)$, called the group of fractional ideals of $L$. Each element (called a fractional ideal) of $I(L)$ has the form $AB^{-1}$ where $A,B$ are ideals of $\mathcal{O}_{L}$. For each $\alpha\in L^{\ast}=L\setminus \{0\}$, $\alpha\mathcal{O}_{L}$ is a fractional ideal, called a principal fractional ideal. From $(\alpha\mathcal{O}_{L})(\beta\mathcal{O}_{L})=\alpha\beta\mathcal{O}_{L}$ and $(\alpha\mathcal{O}_{L})^{-1}=\alpha^{-1}\mathcal{O}_{L}$ we know that the set $P(L)$ of all principal fractional ideals is a subfield of $I(L)$.

\vspace{2 ex}
FACT(4). The group $C(L)=\frac{I(L)}{P(L)}$ is finite and called the (ideal) class group of $L$. $h(L)=|C(L)|$ is called the class number of $L$. An element $[A]$ of $C(L)$ is called the ideal class of a fractional ideal $A$. Therefore $[A]^{h(L)}=1$ for any ideal class $[A]$ in $C(L)$.

\vspace{2 ex}
FACT(5). Let $K=\mathbb{Q}(\zeta_{m})$ where $2\nmid m\geq 3$, and $G=Gal(K/\mathbb{Q})=\{\sigma_{a}:a\in Z_{m}^{\ast}\}$. Let $H=\langle \sigma_{2} \rangle$ be the cyclic subgroup of $G$ generated by $\sigma_{2}$ and $M$ be the subfield of $K$ corresponding to $H$ by Galois Theory. Namely, for $\alpha\in K$, we have $\alpha\in M$ if and only if $\sigma_{2}(\alpha)=\alpha$. $H$ and $M$ are called the decomposition group and decomposition field of 2 in $K$, respectively. We have $$[K:M]=|H|=f, [M:\mathbb{Q}]=[G:H]=g=\frac{\varphi(m)}{f}$$ where $[K:M]=dim_{M}K,$ the dimension of $K$ as a vector space over $M$, $[G:H]=\big|G/H\big|$ and $f$ is order of $2$ in $\mathbb{Z}_{m}^{\ast}$, that is, $f$ is the least positive integer such that $2^{f} \equiv 1\pmod {m}$.

$2\mathcal{O}_{M}$ is an ideal of $\mathcal{O}_{M}$ and $$2\mathcal{O}_{M}=\mathfrak{p}_{1}\mathfrak{p}_{2}\cdots \mathfrak{p}_{g}$$ where $\mathfrak{p}_{i} (1\leq i \leq g)$ are distinct prime ideals of $\mathcal{O}_{M}$. For each $i (1\leq i \leq g),$ $\mathfrak{p}_{i}\mathcal{O}_{K}=P_{i}$ is a prime ideal of $\mathcal{O}_{K}$. Therefore we have the decomposition of 2 in $\mathcal{O}_{K}$ as $$2\mathcal{O}_{K}=P_{1}P_{2}\cdots P_{g}$$ where $P_{1},\ldots, P_{g}$ are distinct prime ideals of $\mathcal{O}_{K}$. From $\mathfrak{p}_{i}\subseteq \mathcal{O}_{M}$ we know that $\sigma_{2}(P_{i})=\sigma_{2}(\mathfrak{p}_{i}\mathcal{O}_{K})=\sigma_{2}(\mathfrak{p}_{i})\mathcal{O}_{K}=\mathfrak{p}_{i}\mathcal{O}_{K}=P_{i} \quad(1\leq i \leq g)$.

\vspace{2 ex}
FACT(6). (decomposition law in imaginary quadratic fields) Let $d=p_{1}\cdots p_{s}$ where $p_{1}, \ldots, p_{s}$ are distinct odd prime numbers, $s\geq1$, $K=\mathbb{Q}(\sqrt{-d})$. Then

(I). $\mathcal{O}_{K}=\{A+B\sqrt{-d}: A,B\in \mathbb{Z}\}$, if $d \equiv 1\pmod 4$,

$\mathcal{O}_{K}=\{\frac{1}{2}(A+B\sqrt{-d}): A,B\in \mathbb{Z}, A\equiv B\pmod 2\}$, if $d \equiv 3\pmod 4$;

(II).
\begin{gather*}
2\mathcal{O}_{K}=\left\{\begin{array}{llllll}
P^{2}, &if d \equiv 1\pmod 4\\
P\bar{P}, P\neq \bar{P}, &if d \equiv 7\pmod 8\\
P,&if d \equiv 3\pmod 8;
\end{array}\right.
\end{gather*}

(III). For any odd prime number $p$,
\begin{gather*}
p\mathcal{O}_{K}=\left\{\begin{array}{llllll}
P^{2}, &if p|d\\
P\bar{P}, P\neq \bar{P}, &if p\nmid d \mbox{ and} \big(\frac{-d}{p}\big)=1\\
P,&if p\nmid d \mbox{ and} \big(\frac{-d}{p}\big)=-1
\end{array}\right.
\end{gather*}
where $P$ denotes a prime ideal of $\mathcal{O}_{K}$ and $\bar{P}=\{\bar{\alpha}:\alpha \in P\}$.

After above preparation, we prove the following result.

\begin{theorem}\label{thm:self}
Let $m$ and $n$ be odd positive integers. Assume that the following semiprimitive condition is satisfied.

$(\ast)$ there exists a positive integer $l$ such that
\begin{align}
2^{l}\equiv-1\pmod m. \label{condition}
\end{align}
Then there is no GBF with type $\{m,n\}$ and $\{2m,n\}$ for any odd integer $n\geq1$.
\end{theorem}

\begin{proof}
Suppose that there exists a GBF $f:Z_{2}^{n} \rightarrow Z_{2m}$, then
$\alpha=W_f(0)=\sum_{x\in Z_{2}^{n}}\zeta_{2m}^{f(x)}\in\mathcal{O}_K$
where $K=\mathbb{Q}(\zeta_{2m})=\mathbb{Q}(\zeta_{m}), \mathcal{O}_K=\mathbb{Z}[\zeta_{m}]$ by FACT(2), and $\alpha\bar{\alpha}=\big|W_f(0)\big|^{2}=2^{n}$ since $f$ is a GBF. Let $l$ be the least positive integer satisfying (\ref{condition}). Then the order of 2 in $Z_{m}^{\ast}$ is $f=2l$ and, by FACT(5), we have
$$2\mathcal{O}_{K}=P_{1}\cdots P_{g}, g=\frac{\varphi(m)}{f}$$ where
$P_{i} (1\leq i\leq g)$ are distinct prime ideals of $\mathcal{O}_{K}$, and
\begin{align}
2^{n}\mathcal{O}_{K}=(P_{1}\cdots P_{g})^{n}.\label{ideal}
\end{align}
On the other hand, from $\alpha\mid \alpha\bar{\alpha}=2^{n}$ where $\bar{\alpha}$ is the complex conjugate number of $\alpha$, we have $$\alpha\mathcal{O}_{K}=P_{1}^{a_{1}}\cdots P_{g}^{a_{g}} \quad(a_{i}\geq0).$$ By FACT(5), we know that $\sigma_{2}(P_{i})=P_{i}$ so that $\sigma_{-1}(P_{i})=\sigma_{2^{l}}(P_{i})=P_{i} (1\leq i\leq g)$. Therefore
$$\sigma_{-1}(\alpha)\mathcal{O}_{K}=\sigma_{-1}(P_{1})^{a_{1}}\cdots \sigma_{-1}(P_{g})^{a_{g}}=P_{1}^{a_{1}}\cdots P_{g}^{a_{g}}=\alpha\mathcal{O}_{K}.$$
But $\alpha=\sum_{i=0}^{\varphi(m)-1}a_{i}\zeta_{m}^{i} (a_{i}\in \mathbb{Z})$ by FACT(2), so we get
$$\sigma_{-1}(\alpha)=\sum_{i=0}^{\varphi(m)-1}a_{i}(\zeta_{m}^{-1})^{i}=\bar{\alpha}.$$
Therefore
\begin{equation}\label{ideal2}
\alpha\bar{\alpha}\mathcal{O}_{K}=(\alpha\mathcal{O}_{K})(\sigma_{-1}(\alpha)\mathcal{O}_{K})=P_{1}^{2a_{1}}\cdots P_{g}^{2a_{g}}.
\end{equation}
From uniqueness of the decompositions (\ref{ideal}) and (\ref{ideal2}) of $\alpha\bar{\alpha}=2^{n}$, we have $2a_{1}=n$ which contradicts to $2\nmid n$.
Therefore there is no GBF with type $\{2m,n\}$ (and type $\{m,n\}$ by Lemma \ref{Lemmaobservation}(2)) for any odd integer $n\geq1$. This completes the proof of Theorem \ref{thm:self}.
\end{proof}

Now we present an explicit description on odd integers $m\geq3$ satisfying the semiprimitive condition (\ref{condition}) (for $m=1$, this condition is true automatically). For each integer $d\neq 0$, let $d=2^{r}d'$ where $r\geq0$ and $2\nmid d'.$ We denote $r$ by $V_{2}(d)$.

\begin{theorem}\label{thm:condition}
Let $m=p_{1}^{a_{1}}\cdots p_{s}^{a_{s}} (s\geq1)$ where $p_{i} (1\leq i \leq s)$ are distinct odd prime numbers. Let $d_{i}$ be the order of $2$ in $Z_{p_{i}}^{\ast}$. Then there exists $l\geq 1$ such that $2^{l}\equiv-1\pmod m$ if and only if $V_{2}(d_{i}) (1\leq i \leq s)$ are the same integer $r\geq 1$.
\end{theorem}
\begin{proof}
Suppose that there exists integer $l\geq 1$ such that $2^{l}\equiv-1\pmod m$. Let $l$ be the least positive integer satisfying this condition. Then the order of 2 in $Z_{m}^{\ast}$ is $2l$. From $2^{l}\equiv-1\pmod {p_{i}^{a_{i}}}$ we know that there exists the least positive integer $l_{i}$ such that $2^{l_{i}}\equiv-1\pmod {p_{i}^{a_{i}}}$, so that the order of 2 in $Z_{p_{i}^{a_{i}}}^{\ast}$ is $2l_{i}$. By Chinese Remainder Theory (CRT) we know that $2l=LCM\{2l_{1},\ldots, 2l_{s}\}$ so that $l=LCM\{l_{1},\ldots, l_{s}\}$ and $V_{2}(l)=max\{V_{2}(l_{1}),\ldots,V_{2}(l_{s})\}$. If $V_{2}(l)>V_{2}(l_{i})$ for some $i$, then $l_{i}\mid \frac{l}{2}$ and $2^{l}\equiv {(2^{\frac{l}{2}})^{2}}\equiv (\pm1)^{2}\equiv 1\pmod {p_{i}^{a_{i}}}$ which is contradicts to $2^{l}\equiv-1\pmod {p_{i}^{a_{i}}}$. Therefore $V_{2}(l_{1})=\cdots=V_{2}(l_{s})\geq 0$. By Euler Theorem, $2l_{i}\mid \varphi(p_{i}^{a_{i}})=p_{i}^{a_{i}-1}(p_{i}-1)$. So that $2l_{i}=b_{i}p_{i}^{c_{i}}$ where $b_{i}\mid p_{i}-1$ and $0\leq c_{i}\leq a_{i}-1$. From elementary number theory we know that the order of $2$ in $Z_{p_{i}}^{\ast}$ is $d_{i}=b_{i}$. Therefore $V_{2}(d_{i})=V_{2}(b_{i})=V_{2}(2l_{i})=V_{2}(l_{i})+1 (1\leq i \leq s)$  are the same positive integer.

On the other hand, suppose that $d_{i}$ be the order of $2$ in $Z_{p_{i}}^{\ast}$ and $V_{2}(d_{1})=\cdots=V_{2}(d_{s})=r\geq1$. Then $d_{i}=2^{r}e_{i}, 2\nmid e_{i}$ and the order of $2$ in $Z_{p_{i}^{a_{i}}}^{\ast}$ is $2l_{i}=d_{i}p_{i}^{c_{i}}$ for some $0\leq c_{i}\leq  a_{i}-1$. Therefore $2^{l_{i}}=2^{2^{r-1}e_{i}p_{i}^{c_{i}}}\equiv-1\pmod {p_{i}^{a_{i}}}$ where $e_{i}p_{i}^{c_{i}}$ is odd $(1\leq i \leq s)$. We take an odd integer $N$ such that $e_{i}p_{i}^{c_{i}}\mid N$ for all $i$, then $h_{i}=\frac{N}{e_{i}p_{i}^{c_{i}}} (1\leq i \leq s)$ are odd integers, so that
$$2^{2^{r-1}N}=(2^{2^{r-1}e_{i}p_{i}^{c_{i}}})^{h_{i}}\equiv(-1)^{h_{i}}\equiv-1\pmod {p_{i}^{a_{i}}} (1\leq i \leq s).$$
By CRT we get $2^{2^{r-1}N}\equiv-1\pmod {m}.$ This completes the proof of Theorem \ref{thm:condition}.
\end{proof}

For an odd prime number $p$, let $d_{p}$ be the order of $2$ in $Z_{p}^{\ast}$. Now we determine $r_{p}=V_{2}(d_{p})$. If $p\equiv7\pmod 8$, then $2^{\frac{p-1}{2}}\equiv \big(\frac{2}{p}\big)=1 \pmod p$ where $ \big(\frac{2}{p}\big)$ is the Legendre symbol. Since $2\nmid \frac{p-1}{2}$ and $d_{p}\mid \frac{p-1}{2}$, we know that $r_{p}=0$. If $p\equiv3\pmod 8$, then $2^{\frac{p-1}{2}}\equiv \big(\frac{2}{p}\big)=-1 \pmod p$. Therefore $d_{p}=2s$ where $s\mid \frac{p-1}{2}$. Since $\frac{p-1}{2}$ is odd, so that $s$ is odd and $r_{p}=V_{2}(d_{p})=V_{2}(2)=1.$ If $p\equiv5\pmod 8$, then
$2^{\frac{p-1}{2}}\equiv \big(\frac{2}{p}\big)=-1 \pmod p$ and $\frac{p-1}{2}=2s, 2\nmid s$. Thus $d_{p}=4s', s'\mid s$, so that $r_{p}=V_{2}(d_{p})=2$. At last, for $p\equiv1\pmod 8$, Table \ref{tab:1} below shows that $r_{p}$ can be different integers.

Let $2\nmid m=p_{1}^{a_{1}}\cdots p_{s}^{a_{s}}$. By Theorem \ref{thm:self} and Theorem \ref{thm:condition} we know that if $r_{p_{i}} (1\leq i \leq s)$ are the same number $r\geq1$, then there is no GBF with type $\{m,n\}$ and $\{2m,n\}$ for any odd integer $n\geq1$. Therefore we get the following result.

\begin{corollary}\label{cor:semi}
Let $m=p_{1}^{a_{1}}\cdots p_{s}^{a_{s}}$ where $p_{i} (1\leq i \leq s)$ are distinct odd prime numbers and $a_{i}\geq1 (1\leq i \leq s)$. Then there is no GBF with type $\{m,n\}$ and $\{2m,n\}$ for all odd integers $n\geq1$ provided one of the following conditions is satisfied.

$(I)(r=1)$: For each $i (1\leq i \leq s), p_{i}\equiv3\pmod 8$, or $p_{i}\equiv1\pmod 8$ with $r_{p_{i}}=1$;

$(II)(r=2)$: For each $i (1\leq i \leq s), p_{i}\equiv5\pmod 8$, or $p_{i}\equiv1\pmod 8$ with  $r_{p_{i}}=2$;

$(III)(r\geq3)$: For each $i (1\leq i \leq s),$ $p_{i}\equiv1\pmod 8$ and $r_{p_{i}}$ are the same number $r\geq3$.
\end{corollary}

\begin{remark}
At the end of this section we explain the meaning of the semiprimitive condition (\ref{condition}) in algebraic number theory. Let $K=\mathbb{Q}(\zeta_{m}), 2\nmid m$, and $M$ be the decomposition field of $2$ in $K$. The decomposition group of $2$ is the subgroup $\langle\sigma_{2} \rangle$ of $G=Gal(K/\mathbb{Q})=\{\sigma_{a}:a\in Z_{m}^{\ast}\}$. Then

there exists $l\geq1$ such that $2^{l}\equiv-1\pmod m$ $ \Leftrightarrow$ $\sigma_{-1}(=\sigma_{2}^{l})\in \langle\sigma_{2} \rangle \Leftrightarrow$ for each $\alpha \in M, \bar{\alpha}(=\sigma_{-1}(\alpha))=\alpha \Leftrightarrow M\subseteq \mathbb{R}$ (the field of real numbers).

Namely, the semiprimitive condition (\ref{condition}) equivalents to that $M$ is a real number field. In next section we will show several new nonexistence results on GBF where the field $M$ is not real.
\end{remark}

\begin{table}\caption{The values of $r_{p}$ for prime numbers $p\equiv1\pmod 8$}
\label{tab:1}
\begin{tabular}[t]{l|cccccccccccc}
\hline
$p\equiv1\pmod 8$ & 17 & 41&73&89&97&113&137&193&257&1553&1777&65537\\
\hline
 $d_{p}$ &8&20&9&11&48&28&68&96&16&194&74&32 \\
\hline
 $r_{p}=V_{2}(d_{p})$ & 3 &2&0&0&3&2&2&5&4&1&1&5\\
 \hline
\end{tabular}
\end{table}

\section{Nonexistence Results (III): Field-Descent}\label{sec:5}
In this section we use the field-descent method given in \cite{F} to show some nonexistence results on GBFs with type $\{m,n\}$ for odd integer $n\geq1$. The method is based on the following result.

\begin{lemma}\cite{F}\label{KD}
Let $K=\mathbb{Q}(\zeta_{m})$ and $M$ be the decomposition field of $2$ in $K$.   If there exists $\alpha \in \mathcal{O}_{K}$ such that $\alpha\bar{\alpha}=2^{n}$, then there is $\beta\in{\mathcal{O}}_{K}$ such that $\beta^{2}\in {\mathcal{O}}_{M}$ and $\beta\bar{\beta}=2^{n}$. Moreover, if $[K:M]$ is odd, then $\beta\in{\mathcal{O}}_{M}$.
\end{lemma}

\begin{proof} Lemma \ref{KD} has been proved as Lemma 2.2 in \cite{F}. Here we copy this proof for convenience of readers. Suppose that $\alpha{\mathcal{O}}_{K}=P_{1}\cdots P_{s}$ where $P_{i} (1\leq i\leq s)$ are prime ideals of ${\mathcal{O}}_{K}$. From $\alpha\bar{\alpha}=2^{n}$ we know that $P_{i}$ is a prime factor of $2{\mathcal{O}}_{K}$, so that $\sigma_{2}(P_{i})=P_{i} (1\leq i\leq s)$ by FACT (5) in Sect. \ref{sec:4}. Therefore
$$\sigma_{2}(\alpha){\mathcal{O}}_{K}=\sigma_{2}(P_{1})\cdots \sigma_{2}(P_{s})=P_{1}\cdots P_{s}=\alpha{\mathcal{O}}_{K}$$ which implies that $\sigma_{2}(\alpha)=\alpha \epsilon$ where $\epsilon \in U_{K}$ (the unit group of ${\mathcal{O}}_{K}$). For each $\sigma\in G=Gal(K/\mathbb{Q})$ we have
$$\sigma(\alpha)\overline{\sigma(\alpha)}=\sigma(\alpha)\sigma_{-1}(\sigma(\alpha))=\sigma(\alpha)\sigma(\sigma_{-1}(\alpha))=
\sigma(\alpha)\sigma(\bar{\alpha})=\sigma(\alpha\bar{\alpha})=\sigma(2^{n})=2^{n},$$
$$\sigma\sigma_{2}(\alpha)=\sigma(\alpha\epsilon)=\sigma(\alpha)\sigma(\epsilon).$$
Thus $$2^{n}=\sigma\sigma_{2}(\alpha)\cdot\sigma\sigma_{2}(\bar{\alpha})=
\sigma(\alpha)\sigma(\epsilon)\overline{\sigma(\alpha)} \overline{\sigma(\epsilon)}=2^{n}|\sigma(\epsilon)|^{2}.$$
Therefore $|\sigma(\epsilon)|=1$ for any $\sigma\in G$. By FACT (2), $\epsilon\in W_{K}.$ Namely, $\epsilon=\pm\delta$ and $\delta=\zeta_{m}^{i}$ for some $i$. Let $\beta=\alpha\delta^{-1},$ then $\beta\bar{\beta}=\alpha\bar{\alpha}=2^{n}$ and $$\sigma_{2}(\beta)=\sigma_{2}(\alpha)\sigma_{2}(\delta)^{-1}=\alpha\epsilon\delta^{-2}=\pm\alpha\delta^{-1}=\pm\beta.$$
Therefore $\sigma_{2}(\beta^{2})=\beta^{2}$ which means that $\beta^{2}\in{\mathcal{O}}_{M}.$
Moreover, if $[K:M]$ is odd, then $M$ has no quadratic extension in $K$ so that $\beta\in {\mathcal{O}}_{M}$.
\end{proof}

From Corollary \ref{cor:semi} we know that for an odd prime number $p$, there is no GBF with type $\{2p^{l},n\}$ for any $l\geq1$ and odd $n\geq1$ if $p\equiv 3,5\pmod 8$, or $p\equiv1\pmod 8$ and $2|f$ where $f$ is the odd of 2 in $Z_{p}^{^{*}}.$ For the remain (non-semiprimitive) case $p\equiv 7\pmod 8$, we have the following result.

\begin{theorem}\label{thm:p=7}
Let $m=p^{l}$ where $l\geq1$ and $p\equiv7\pmod8$ is a prime number. Let $f$ be the order of $2$ in $Z_{p}^{^{*}}$ (from $\big(\frac{2}{p}\big)=1$ we know that $2\nmid f$),
$g=\frac{\varphi(p^{l})}{f}$ (is even) and $s=\frac{g}{2}$ (is odd since $\varphi(p^{l})=(p-1)p^{l-1}\equiv2\pmod4$). Let $r$ be the least positive odd integer such that $x^{2}+py^{2}=2^{r+2}$ has a solution $(x,y), x,y\in \mathbb{Z}$. Then
there is no GBF with type $\{2p^{l},n\}$ for any odd integer $1\leq n<\frac{r}{s}$.
\end{theorem}

\begin{proof}

Suppose there exists a GBF $f$ with type $\{2m,n\}$ , then $\alpha=W_{f}(y)\in\mathcal{O}_K (K=\mathbb{Q}(\zeta_{m}))$ for any $y\in Z_{2}^{n}$ and $\alpha\bar{\alpha}=2^n$. By Lemma \ref{KD}, there exists $\beta\in{\mathcal{O}}_{K}$ such that $\beta^{2}\in{\mathcal{O}}_{M}$ and $\beta\bar{\beta}=2^n$ where $M$ is the decomposition field of 2 in $K$. Since $[K:M]=f$ is odd, we know that $\beta\in{\mathcal{O}}_{M}$ and $[M:\mathbb{Q}]=g=2s$. The Galois group $G=Gal(K/\mathbb{Q})$  is isomorphic to the  cyclic group $\mathbb{Z}_{p^{l}}^{*}$. By Galois theory, there exists an unique quadratic field $T$, $\mathbb{Q}\subset T\subseteq M, [T:\mathbb{Q}]=2$ and $[M:T]=s$. It is well known that $T$ is the imaginary quadratic field $\mathbb{Q}(\sqrt{-p})$ $\big($one of ways for proving this fact is by the quadratic Gauss sum $\sum_{i=1}^{p-1}\big(\frac{i}{p}\big)\zeta_{p}^{i}=\sqrt{-p}$, so that $\sqrt{-p}\in \mathbb{Q}(\zeta_{p})\subseteq \mathbb{Q}(\zeta_{m})=K$ and $T=\mathbb{Q}(\sqrt{-p})\big)$.

\begin{tikzpicture}
\draw[-](0,0)--(0,.5);
\node [above] at (0,.5){$K=\mathbb{Q}(\zeta_{m})$};
\node[below] at (0,0){$M$};
\node[left] at (0,.25){$f$};
\draw[-](0,-.4)--(0,-1.1);
\node[below] at (0,-1.1){$T=\mathbb{Q}(\sqrt{-p})$};
\node[left] at (0,-.8){$s$};
\draw[-](0,-1.6)--(0,-2.1);
\node[below] at (0,-2.1){$\mathbb{Q}$};
\node[left] at (0,-1.9){$2$};
\draw[-](2,0)--(2,.5);
\node [above] at (2,.5){$(1)$};
\node[below] at (2,0){$\langle \sigma_{2} \rangle$};
\draw[-](2,-.4)--(2,-2.1);
\node[below] at (2,-2.1){$G$};
\end{tikzpicture}

Let $N_{M/T}:M\longrightarrow T$ be the norm mapping and $\gamma=N_{M/T}(\beta)\in\mathcal{O}_{T}$, then
$\gamma\bar{\gamma}=N_{M/T}(\beta)N_{M/T}(\bar{\beta})=N_{M/T}(\beta\bar{\beta})=N_{M/T}(2^{n})=2^{ns}$ where $ns$ is odd. It is known that $\gamma \in {\mathcal{O}}_{T}$ can be expressed as $\gamma=\frac{1}{2}(A+B\sqrt{-p})$ where $A,B\in \mathbb{Z}, A \equiv B\pmod 2$ (FACT (6) in Sect. \ref{sec:4}).
Therefore  $2^{ns}=\gamma\bar{\gamma}=\frac{1}{4}(A^{2}+pB^{2})$. By the definition of $r$ and $2\nmid ns$ we get $ns\geq r$. Therefore there is no GBF with type $\{2p^{l},n\}$ if $2\nmid n,n<\frac{r}{s}$. This completes the proof of Theorem \ref{thm:p=7}.
\end{proof}

\begin{remark}\label{rem:p7}
(A). Let $p$ be an odd prime number. For any $l\geq1$, we denote $f_{l}$ be the order of 2 modulo $p^{l}$, $g_{l}=\frac{\varphi(p^{l})}{f_{l}}$. It is easy to see that if $2^{p-1}\not\equiv1\pmod {p^{2}}$ then
$f_{l}=p^{l-1}f_{1}$, so that $g_{l}=\frac{\varphi(p^{l})}{f_{l}}=\frac{\varphi(p)p^{l-1}}{f_{1}p^{l-1}}=g_{1}$ for all  $l\geq 2$. It is known that $2^{p-1}\not\equiv1\pmod {p^{2}}$ for all odd prime numbers $p<6\cdot10^{9}$ except $p=1093$ and $3511$. Thus for such $p$ it is enough to compute $g=g_{1}$.

(B). The definition of $r$ is elementary. Now we prove existence of $r$ and explain its meaning in algebraic number theory. Since $r$ is the least odd integer such that $x^{2}+py^{2}=2^{r+2}$ has a solution  $(x,y)=(A,B)\in \mathbb{Z}^{2}$, we know that both $A$ and $B$ are odd, so that $\delta=\frac{1}{2}(A+B\sqrt{-p})\in \mathcal{O}_{T} (T=\mathbb{Q}(\sqrt{-p}))$ and $\delta\bar{\delta}=2^{r}$. From $p\equiv7\pmod8$ we know that $2\mathcal{O}_{T}=P\bar{P}$ where $P$ and $\bar{P}=\sigma_{-1}(P)$ are distinct prime ideals of $\mathcal{O}_{T}$ (FACT (6) of Sect. \ref{sec:4}). Let $[P]$ be the ideal class of $P$ in the class group $C(T)$. We have $1=[2 \mathcal{O}_{T}]=[P][\bar{P}]$ so that $[\bar{P}]=[P]^{-1}$. From $\delta\bar{\delta}\mathcal{O}_{T}=2^{r}\mathcal{O}_{T}=P^{r}\bar{P}^{r}$ and the minimum property of $r$ we get $\delta\mathcal{O}_{T}=P^{r}$ or $\bar{P}^{r}$ and $r$ is the order of $[P]$ in $C(T)$. Therefore $r|h(T)$ (the class number of $T$). By the Gauss genus theory, $h(T)$ is odd for $T=\mathbb{Q}(\sqrt{-p})$ and $p\equiv7\pmod8$. On the other hand, we have $2^{r+2}=A^{2}+pB^{2}>p$ which implies that $r>\frac{\log p}{\log 2}-2$. Particular we have $r\geq3$ if $p\equiv7\pmod8$ and $p\geq23$. Thus $r=h(T)$ if $h(T)$ is a prime number.
\end{remark}

\begin{example}
There are $11$ primes satisfying $p\equiv7\pmod8$ within $200$. The Table \ref{tab:2} presents the values of $s, h(L)$ and $r$ for $L=\mathbb{Q}(\sqrt{-p})$.
\begin{table}\caption{The values of $s$ and $r$ for prime numbers $p\equiv7\pmod 8$}
\label{tab:2}
\begin{tabular}[t]{l|ccccccccccc}
\hline
   $p$ & 7 & 23 & 31 & 47 & 71 & 79 & 103 & 127 & 151 & 191 & 199 \\
\hline
   $s$ & 1 & 1 & 3 & 1 & 1 & 1 & 1 & 9 & 5 & 1 & 1   \\
\hline
  $h(L)$ & 1 & 3 & 3 & 5 & 7 & 5 & 5 & 5 & 7 & 13 & 9   \\
\hline
 $r$ & 1 & 3 & 3 & 5 & 7 & 5 & 5 & 5 & 7 & 13 & 9   \\
 \hline
\end{tabular}
\end{table}

For $p=47, 79$ and $103$, we have $s=1$ and $r=5$. By Theorem \ref{thm:p=7}, there is no GBF with type $\{2\cdot p^{l},1\}$ and $\{2\cdot p^{l},3\}$ for all $l\geq1$. For $p=199$, We have $s=1$ and $r=9$. There is no GBF with type $\{2\cdot 199^{l},n\}$ for $n=1,3,5,7$ and all $l\geq1$. Similarly, there is no  GBF with type $\{2\cdot 191^{l},n\}$ for $n=1,3,5,7,9,11$ and all $l\geq1$.
\end{example}

The following two corollaries are direct consequences of Theorem \ref{thm:p=7}.

\begin{corollary}
Suppose that $p\equiv7\pmod8$ is a prime number, $2^{p-1}\not\equiv1\pmod{p^{2}}$, and
 the order of $2$ modulo $p$ is $\frac{p-1}{2}$. Then there is no GBF with type $\{2p^{l},n\}$ for all $l\geq1$ and $2\nmid n<r$ where $r$ is defined in Theorem \ref{thm:p=7}.
\end{corollary}

\begin{corollary}
Suppose that $p\equiv7\pmod8$ is a prime number, $p\geq23$, $2^{p-1}\not\equiv1\pmod {p^2}$
and the ideal class number $h(T)$ of $T=\mathbb{Q}(\sqrt{-p})$ be a prime number. Then there is no GBF with type $\{2p^{l},n\}$ for all $l\geq1$ and $2\nmid n<\frac{h(T)}{s}$ where $s=\frac{p-1}{2f}$ and $f$ is the order of $2$ modulo $p$.
\end{corollary}

From now on we assume that $m=p_{1}^{a_{1}}p_{2}^{a_{2}}$ where $p_{1}, p_{2}$ are distinct odd prime numbers, $a_{1}, a_{2}\geq1$. Let $f_{1}, f_{2}$ and $f$ be the order of $2$ in $Z_{p_{1}^{a_{1}}}^{\ast}, Z_{p_{2}^{a_{2}}}^{*}$ and $Z_{m}^{*}$ respectively. Then $f=LCM\{f_{1},f_{2}\}$. Let $M_{1}, M_{2}$ and $M$ be the decomposition fields of 2 in $K_{1}=\mathbb{Q}(\zeta_{p_{1}^{a_{1}}}), K_{2}=\mathbb{Q}(\zeta_{p_{2}^{a_{2}}})$ and $K=\mathbb{Q}(\zeta_{m})=K_{1}K_{2}$ respectively. Then $g_{i}=[M_{i}:\mathbb{Q}]=\frac{\varphi(p_{i}^{a_{i}})}{f_{i}} (i=1,2)$ and
$$g=[M:\mathbb{Q}]=\frac{\varphi(m)}{f}=\frac{\varphi(p_{1}^{a_{1}})\varphi(p_{2}^{a_{2}})}{LCM\{f_{1},f_{2}\}}
=\frac{\varphi(p_{1}^{a_{1}})}{f_{1}}\frac{\varphi(p_{2}^{a_{2}})}{f_{2}}\cdot GCD\{f_{1},f_{2}\}=g_{1}g_{2}\cdot GCD\{f_{1},f_{2}\}.$$

\begin{theorem}\label{thm:735}
Suppose that $p_{1}\equiv7\pmod8$, $p_{2}\equiv 3,5\pmod8$ and $m=p_{1}^{a_{1}}p_{2}^{a_{2}}$ $(a_{1}, a_{2}\geq1)$. Then $g=2s$ and $s$ is odd. Let $r_{1}\geq1$ be the least odd integer such that $x^{2}+p_{1}y^{2}=2^{r_{1}+2}$ has a solution $(x,y), x,y \in \mathbb{Z}$. Let $r_{2}\geq1$ be the least odd integer such that $x^{2}+p_{1}y^{2}=2^{r_{2}+2}p_{2}$ has a solution $(x,y), x,y \in \mathbb{Z}$. (If there is no solution $(x,y)\in \mathbb{Z}^{2}$ of $x^{2}+p_{1}y^{2}=2^{r_{2}+2}p_{2}$ for any odd $r_{2}\geq1$, we assume that $r_{2}=+\infty$). Let $r=min\{r_{1},r_{2}\}$. Then there is no GBF with type $\{2m,n\} (\mbox{and }\{m,n\})$ provided one of the following conditions is satisfied.

(I). $2\nmid n<\frac{r_{1}}{s}$ and $\big(\frac{-p_{1}}{p_{2}}\big)=-1$;

(II). $2\nmid n<\frac{r}{s}$ and $\big(\frac{-p_{1}}{p_{2}}\big)=1$.
\end{theorem}

\begin{proof}

Firstly we consider the case $p_{1}\equiv7\pmod8$ and $p_{2}\equiv 3\pmod8$. We know that $$f_{1}\mbox{\ is odd and } f_{1}\mid\frac{1}{2}\varphi(p_{1}^{a_{1}})=\frac{p_{1}-1}{2}p_{1}^{a_{1}-1},$$
$$f_{2}=2f_{2}', f_{2}' \mbox{ is odd and } f_{2}'\mid\frac{1}{2}\varphi(p_{2}^{a_{2}})=\frac{p_{2}-1}{2}p_{2}^{a_{2}-1}.$$
Therefore $g_{1}=\frac{\varphi(p_{1}^{a_{1}})}{f_{1}}\equiv2\pmod4, g_{2}=\frac{\varphi(p_{2}^{a_{2}})}{f_{2}}$ is odd, and $g=g_{1}g_{2}\cdot GCD\{f_{1},f_{2}'\}=2s, 2\nmid s$. The imaginary quadratic field $T=\mathbb{Q}(\sqrt{-p_{1}})$ is a subfield of $K_{1}$. Since $M\cap K_{i}=M_{i} (i=1,2)$ and $[M_{1}:\mathbb{Q}]=g_{1}$ is even, $[M_{2}:\mathbb{Q}]=g_{2}$ is odd, we know that $T\subseteq M_{1}\subseteq M, \sqrt{-p_{2}}\not\in M.$ Therefore $L=M(\sqrt{-p_{2}})$ is a quadratic extension of $M$ in $K$. Moreover, since $Gal(K/M)$ is the cyclic group generated by $\sigma_{2}$ and $[K:M]=\frac{\varphi(m)}{g}$ is even, we know that $L$ is the unique quadratic extension of $M$ in $K$.

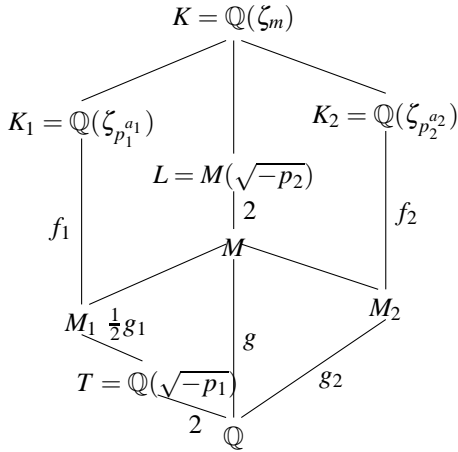
\begin{tikzpicture}
\draw[-](0,0)--(0,1.5);
\node [above] at (0,1.5){$K=\mathbb{Q}(\zeta_{m})$};
\node [below] at(0,0){$L=M(\sqrt{-p_{2}})$};
\draw[-](-.1,1.5)--(-2,.7);
\node[below]at(-2,.7){$K_{1}=\mathbb{Q}(\zeta_{p_{1}^{a_{1}}})$};
\draw[-](.1,1.5)--(2,.8);
\node[below]at(2,.8){$K_{2}=\mathbb{Q}(\zeta_{p_{2}^{a_{2}}})$};
\draw[-](-2,.2)--(-2,-2)
node[left] at (-2,-1){$f_{1}$};
\node[below] at (-2,-2){$M_{1}$};
\draw[-](0,-.5)--(0,-1);
\node[right] at (0,-.75){2};
\node[below] at (0,-1){$M$};
\draw[-](2,.3)--(2,-1.8);
\node[right]at(2,-.8){$f_{2}$};
\node[below]at(2,-1.8) {$M_{2}$};
\draw[-](0,-1.4)--(0,-3.5);
\node[below] at(0,-3.5){$\mathbb{Q}$};
\node[right] at(0,-2.5){$g$};
\draw[-](2,-2.2)--(.1,-3.5);
\node[right] at (1,-3){$g_{2}$};
\draw[-](.1,-1.2)--(1.9,-1.8);
\draw[-](-.1,-1.2)--(-1.9,-2);
\draw[-](-2,-2.4)--(-1.2,-2.75);
\node[below] at(-1,-2.75) {$T=\mathbb{Q}(\sqrt{-p_{1}})$};
\draw[-](-1,-3.2)--(-.1,-3.5);
\node[above]at(-1.4,-2.6){$\frac{1}{2}g_{1}$};
\node[below]at(-.5,-3.35){$2$};
\end{tikzpicture}

Suppose that there exists a GBF with type $\{2m,n\}$, $2\nmid n\geq1$. Then we have $\alpha\in \mathcal{O}_{K}$ such that $\alpha\bar{\alpha}=2^{n}$. By lemma \ref{KD}, there exists $\beta\in \mathcal{O}_{L}$ such that $\beta^{2}\in \mathcal{O}_{M}$ and $\beta\bar{\beta}=2^{n}$. Since $\{1,\sqrt{-p_{2}}\}$ is a basis for extension $L/M,$ we have $\beta=A+B\sqrt{-p_{2}}, A,B\in M$. By $\beta^{2}=A^{2}-p_{2}B^{2}+2AB\sqrt{-p_{2}}\in \mathcal{O}_{M}$ and $\sqrt{-p_{2}}\not\in M$ we get $AB=0$. Therefore $\beta=A\in \mathcal{O}_{L}\cap M=\mathcal{O}_{M}$ or $\beta=B\sqrt{-p_{2}}\in \mathcal{O}_{L}$.

If $\beta=A\in \mathcal{O}_{M},$ then $\gamma=N_{M/T}(\beta)\in \mathcal{O}_{T}$ where $N_{M/T}(\beta)=\Pi_{\sigma\in Gal(M/T)}\sigma(\beta)$ is the norm mapping from $M$ to $T$ (and from $\mathcal{O}_{M}$ to $\mathcal{O}_{T}$). We have $\gamma\bar{\gamma}=N_{M/T}(\beta\bar{\beta})=N_{M/T}(2^{n})=2^{ns}$ where $s=\frac{g}{2}=[M:T]$. By FACT (6) of Sect. \ref{sec:4}, $\gamma=\frac{1}{2}(x+y\sqrt{-p_{1}})$ where $x, y\in \mathbb{Z}$. Then $x^{2}+p_{1}y^{2}=4\gamma\bar{\gamma}=2^{ns+2}$. From $2\nmid ns$ and the definition of $r_{1}$ we get $ns\geq r_{1}$.

If $\beta=B\sqrt{-p_{2}}\in \mathcal{O}_{L}$ and $B\in M$, consider the element $\gamma=\prod_{\sigma\in Gal(L/T)/Gal(L/M)}\sigma(\beta)=\prod_{\sigma\in Gal(M/T)}\sigma(B)\cdot\prod_{\sigma\in Gal(L/T)/Gal(L/M)}\sigma(\sqrt{-p_{2}})=N_{M/T}(B)\cdot (\pm\sqrt{-p_{2}})^{s}$ where $N_{M/T}(B) \in T$. Let $\delta=N_{M/T}(B)p_{2}^{\frac{s+1}{2}}\in T\cap\mathcal{O}_{L}=\mathcal{O}_{T}$. Then $\delta=\frac{1}{2}(x+y\sqrt{-p_{1}})$, $x,y \in \mathbb{Z}$ and
$x^{2}+p_{1}y^{2}=4\delta\bar{\delta}=4N_{M/T}(B\bar{B})p_{2}^{s+1}=4\gamma\bar{\gamma}p_{2}=2^{ns+2}p_{2}.$ Therefore $p_{2}\nmid xy$ and $x^{2}\equiv -p_{1}y^{2}\pmod {p_{2}}$. We get $\big(\frac{-p_{1}}{p_{2}}\big)=1$. This means that if  $\big(\frac{-p_{1}}{p_{2}}\big)=-1$ then $\beta=B\sqrt{-p_{2}}$ is impossible. If $\big(\frac{-p_{1}}{p_{2}}\big)=1$, from $x^{2}+p_{1}y^{2}=2^{ns+2}p_{2}$ and the definition of $r_{2}$ we get $ns\geq r_{2}$.

In summary, if there is a GBF with type $\{2m,n\} (2\nmid n),$ then $n\geq\frac{r_{1}}{s}$ when $\big(\frac{-p_{1}}{p_{2}}\big)=-1$ and $n\geq min\{\frac{r_{1}}{s},\frac{r_{2}}{s}\}=\frac{r}{s}$ when $\big(\frac{-p_{1}}{p_{2}}\big)=1$. This completes the proof of Theorem \ref{thm:735} for the case $p_{1}\equiv7\pmod8$ and $p_{2}\equiv 3\pmod8$.

Now we consider the case $p_{1}\equiv7\pmod8$ and $p_{2}\equiv 5\pmod8$. In this case, $$f_{1}\mbox{\ is odd and } f_{1}\mid\frac{1}{2}\varphi(p_{1}^{a_{1}}),$$
$$f_{2}=4f_{2}', f_{2}' \mbox{ is odd and } f_{2}'\mid\frac{1}{4}\varphi(p_{2}^{a_{2}}).$$
Therefore $g_{1}=\frac{\varphi(p_{1}^{a_{1}})}{f_{1}}\equiv2\pmod4, g_{2}=\frac{\varphi(p_{2}^{a_{2}})}{f_{2}}$ is odd, and $g=g_{1}g_{2}\cdot GCD\{f_{1},f_{2}'\}=2s, 2\nmid s$.
Then we have the same field extension diagram as one in the case $(p_{1},p_{2})\equiv (7,3)\pmod8$ except $L=M(\sqrt{p_{2}})$ for $p_{2}\equiv 5\pmod8$. The conclusion can be derived by the same argument. This completes the proof of Theorem \ref{thm:735}.
\end{proof}

\begin{remark}
We have seen the existence of $r_{1}$ in Remark (B) of Theorem \ref{thm:p=7}. Namely, $2\mathcal{O}_{T}=P\bar{P}$ and $r_{1}$ is the order of the ideal class $[P]$ in the class group $C(T)$. Particularly $r_{1}\mid h(T)=|C(T)|$. By Gauss genus theory, $h(T)$ is odd, and so is $r_{1}$.
\end{remark}

Now we show an explicit method to determine the value of $r_{2}$ in the case $p_{1}\equiv7\pmod8$, $p_{2}\equiv 3,5\pmod8$ and $\big(\frac{-p_{1}}{p_{2}}\big)=1$. By definition, $r_{2}=+\infty$ if the equation $x^{2}+p_{1}y^{2}=2^{l+2}p_{2}$ has no solution $(x,y)\in \mathbb{Z}^{2}$ for any odd integer $l\geq1$. Otherwise $r_{2}$ is the least odd integer such that $x^{2}+p_{1}y^{2}=2^{r_{2}+2}p_{2}$ has a solution $(x,y)\in \mathbb{Z}^{2}$. We claim that if $r_{2}\neq +\infty,$ then $r_{2}\leq r_{1}$. So that $r=min\{r_{1},r_{2}\}=r_{2}$.

Suppose that $r_{2}\neq +\infty,$ so that $x^{2}+p_{1}y^{2}=2^{r_{2}+2}p_{2}$ has a solution $(x,y)=(A,B)\in \mathbb{Z}^{2}$. From the minimal property of $r_{2}$ we know that $2\nmid AB$ so that $\pi=\frac{1}{2}(A+B\sqrt{-p_{1}})\in \mathcal{O}_{T}$. We have $2\mathcal{O}_{T}=P\bar{P}$ and $p_{2}\mathcal{O}_{T}=Q\bar{Q}$ by $\big(\frac{-p_{1}}{p_{2}}\big)=1$ (FACT (6) of Sect. \ref{sec:4}). Therefore $$(\pi\mathcal{O}_{T})(\overline{\pi}\mathcal{O}_{T})=(2^{r_{2}}p_{2})\mathcal{O}_{T}=(P\bar{P})^{r_{2}}Q\bar{Q}.$$
Again, from the minimal property of $r_{2}$ we get $$\pi\mathcal{O}_{T}=(P')^{r_{2}}Q', P'\in \{P, \bar{P}\}, Q'\in \{Q,\bar{Q}\}.$$
Therefore $[P']^{r_{2}}[Q']=1$. Namely, $[Q']$ belongs to the cyclic group of order $r_{1}$ generated by $[P]$. Inversely, suppose that $[P']^{l}[Q']=1$ for some $l$, $0\leq l<r_{1}-1$. If $l$ is odd, then $r_{2}=l\leq r_{1}-1$. If $l$ is even, then $[\bar{P}']^{r_{1}-l}[\bar{Q}']=([P']^{l}[Q'])^{-1}=1$ so that $r_{2}=r_{1}-l\leq r_{1}$.

In summary, if the equation $x^{2}+p_{1}y^{2}=2^{l+2}p_{2}$ has no solution $(x,y)\in \mathbb{Z}^{2}$ for $l=0,1,\ldots, r_{1}-1$, then $r_{2}=+\infty$ and $r=r_{1}$. Otherwise there exists an unique odd $l$ such that $1\leq l \leq r_{1}$ and $x^{2}+p_{1}y^{2}=2^{l+2}p_{2}$ has solution $(x,y)\in \mathbb{Z}^{2}$. Then $r_{2}=l\leq r_{1}$ and $r=min\{r_{1},r_{2}\}=r_{2}$.

\begin{example}
Consider $m=p_{1}^{a_{1}}p_{2}^{a_{2}}$ where $p_{1}=199\equiv7\pmod8$. By the Table \ref{tab:2} we know that $g_{1}=2$ and $f_{1}=\frac{\varphi(p_{1}^{a_{1}})}{g_{1}}=99\cdot 199^{a_{1}-1}$ and $r_{1}=9$.

(1). For $p_{2}=59\equiv3\pmod8, f_{2}=2f_{2}', f_{2}'=29\cdot 59^{a_{2}-1}, g_{2}=1$. Thus $g=g_{1}g_{2}\cdot GCD\{f_{1},f_{2}'\}=2$ and $s=\frac{g}{2}=1$. Since $\big(\frac{-p_{1}}{p_{2}}\big)=\big(\frac{-199}{59}\big)=-1,$ by Theorem \ref{thm:735} we get that for any $a_{1},a_{2}\geq 1,$ and $n=1,3,5,7$, there is no GBF with type $\{2\cdot199^{a_{1}}\cdot59^{a_{2}}, n\}$ and $\{199^{a_{1}}\cdot59^{a_{2}}, n\}$.

(2). For $p_{2}=101\equiv5\pmod8, f_{2}=4f_{2}', f_{2}'=25\cdot 101^{a_{2}-1}, g_{2}=1$. Thus $g=g_{1}g_{2}\cdot GCD\{f_{1},f_{2}'\}=2$ and $s=\frac{g}{2}=1$. Since $\big(\frac{-p_{1}}{p_{2}}\big)=\big(\frac{-199}{101}\big)=-1,$ by Theorem \ref{thm:735} we get that for any $a_{1},a_{2}\geq 1,$ and $n=1,3,5,7$, there is no GBF with type $\{2\cdot199^{a_{1}}\cdot 101^{a_{2}}, n\}$ and $\{199^{a_{1}}\cdot101^{a_{2}}, n\}$.

(3). For $p_{2}=5, f_{2}=4f_{2}', f_{2}'=5^{a_{2}-1}, g_{2}=1$. Thus $g=g_{1}g_{2}\cdot GCD\{f_{1},f_{2}'\}=2$ and $s=\frac{g}{2}=1$. From $11^{2}+199\cdot 1^{2}=320=2^{6}\cdot5$ we get $r_{2}=9-4=5$ (in fact, $21^{2}+199\cdot 1^{2}=2^{7}\cdot5$) and $r=r_{2}=5$. Since $\big(\frac{-p_{1}}{p_{2}}\big)=\big(\frac{-199}{5}\big)=1,$ by Theorem \ref{thm:735} we get that for any $a_{1},a_{2}\geq 1,$ and $n=1,3$, there is no GBF with type $\{2\cdot199^{a_{1}}\cdot 5^{a_{2}}, n\}$ and $\{199^{a_{1}}\cdot 5^{a_{2}}, n\}$.
\end{example}

\begin{theorem}\label{thm:35}
Suppose that $p_{1}\equiv3\pmod8$, $p_{2}\equiv 5\pmod8$ and $m=p_{1}^{a_{1}}p_{2}^{a_{2}}$ $(a_{1}, a_{2}\geq1)$. Then $g=2s$ and $s$ is odd. Let $r\geq1$ be the least odd integer such that $p_{1}x^{2}+p_{2}y^{2}=2^{r+2}$ has a solution $(x,y), x,y \in \mathbb{Z}$. Then there is no GBF with type $\{2m,n\}$ and $\{m,n\}$ provided $n$ is odd and one of the following conditions is satisfied.

(I). $\big(\frac{p_{2}}{p_{1}}\big)=1$; or

(II). $\big(\frac{p_{2}}{p_{1}}\big)=-1$ and $n<\frac{r}{s}$.
\end{theorem}

\begin{proof}
We have $$f_{1}=2f_{1}', 2\nmid f_{1}', g_{1}=\frac{\varphi(p_{1}^{a_{1}})}{f_{1}} \mbox{ is odd},$$
$$f_{2}=4f_{2}', 2\nmid f_{2}', g_{2}=\frac{\varphi(p_{2}^{a_{2}})}{f_{2}} \mbox{ is odd}.$$
Therefore $g=g_{1}g_{2}\cdot GCD\{f_{1},f_{2}\}=2g_{1}g_{2}\cdot GCD\{f_{1}',f_{2}'\}=2s, 2\nmid s$. From $2\nmid g_{1}g_{2}$ we know that both $\mathbb{Q}(\sqrt{-p_{1}})$ and $\mathbb{Q}(\sqrt{p_{2}})$ are not subfields of $M$. Therefore $\sigma_{2}(\sqrt{-p_{1}})=-\sqrt{p_{1}}$ and $\sigma_{2}(\sqrt{p_{2}})=-\sqrt{p_{2}}$, so that $\sigma_{2}(\sqrt{-p_{1}p_{2}})=\sqrt{-p_{1}p_{2}}$ which means that $T=\mathbb{Q}(\sqrt{-p_{1}p_{2}})$ is the quadratic subfield of $M$ and the quadratic extension of $M$ is $L=M(\sqrt{p_{2}}) (=M(\sqrt{-p_{1}}))$. Thus we have the field extension diagram as following.

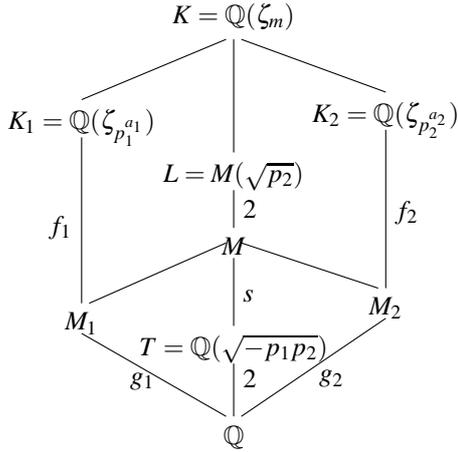
\begin{tikzpicture}

\draw[-](0,0)--(0,1.5);
\node [above] at (0,1.5){$K=\mathbb{Q}(\zeta_{m})$};
\node [below] at(0,0){$L=M(\sqrt{p_{2}})$};
\draw[-](-.1,1.5)--(-2,.7);
\node[below]at(-2,.7){$K_{1}=\mathbb{Q}(\zeta_{p_{1}^{a_{1}}})$};
\draw[-](.1,1.5)--(2,.8);
\node[below]at(2,.8){$K_{2}=\mathbb{Q}(\zeta_{p_{2}^{a_{2}}})$};
\draw[-](-2,.2)--(-2,-2)
node[left] at (-2,-1){$f_{1}$};
\node[below] at (-2,-2){$M_{1}$};
\draw[-](0,-.5)--(0,-1);
\node[right] at (0,-.75){2};
\node[below] at (0,-1){$M$};
\draw[-](0,-1.4)--(0,-2.3);
\node[below] at(0,-2.3) {$T=\mathbb{Q}(\sqrt{-p_{1}p_{2}})$};
\draw[-](2,.3)--(2,-1.8);
\node[right]at(2,-.8){$f_{2}$};
\node[below]at(2,-1.8) {$M_{2}$};
\draw[-](0,-2.8)--(0,-3.5);
\node[below] at(0,-3.5){$\mathbb{Q}$};
\node[right] at(0,-1.9){$s$};
\node[right]at(0,-3){2};
\draw[-](2,-2.2)--(.1,-3.5);
\node[right] at (1,-3){$g_{2}$};
\draw[-](.1,-1.2)--(1.9,-1.8);
\draw[-](-.1,-1.2)--(-1.9,-2);
\draw[-](-2,-2.4)--(-.1,-3.5);
\node[below]at(-1.2,-2.8){$g_{1}$};
\end{tikzpicture}

Suppose that there exists a GBF with type $\{2m,n\}$ and $2\nmid n\geq1$. Then we have $\alpha\in\mathcal{O}_L$ such that $\alpha^{2}\in \mathcal{O}_M$ and $\alpha\bar{\alpha}=2^{n}$. From $\alpha=A+B\sqrt{p_{2}} (A,B\in M)$ and $\alpha^{2}=A^{2}+p_{2}B^{2}+2AB\sqrt{p_{2}}\in\mathcal{O}_M$ we know that $AB=0$. Therefore $\alpha=A\in M\cap \mathcal{O}_L=\mathcal{O}_M$ or $\alpha=B\sqrt{p_{2}}.$ If $\alpha=A\in \mathcal{O}_M$, then $\beta=N_{M/T}(\alpha)\in \mathcal{O}_T$ and $\beta\bar{\beta}=N_{M/T}(\alpha\bar{\alpha})=N_{M/T}(2^{n})=2^{ns}$. By $\beta=\frac{1}{2}(x+y\sqrt{-p_{1}p_{2}}), x,y\in \mathbb{Z},$ we have $x^{2}+p_{1}p_{2}y^{2}=2^{ns+2}$ so that $x^{2}\equiv2^{ns+2}\pmod {p_{1}}$ and then $\big(\frac{2}{p_{1}}\big)=1$ since $2\nmid ns$. But $\big(\frac{2}{p_{1}}\big)=-1$ since $p_{1}\equiv3\pmod8$. We get a contradiction. Therefore $\alpha=B\sqrt{p_{2}}\in \mathcal{O}_L, B\in M$. Let $$\gamma=\prod_{\sigma\in Gal(L/T)/Gal(L/M)}\sigma(\alpha)=N_{M/T}(B)(\pm\sqrt{p_{2}}^{s})\in \mathcal{O}_L$$ where $N_{M/T}(B) \in T$. Then $\delta=N_{M/T}(B)p_{2}^{\frac{s+1}{2}}=\gamma \sqrt{p_{2}}\in \mathcal{O}_L\cap T=\mathcal{O}_T$ and $$\delta\bar{\delta}=N_{M/T}(B\bar{B})p_{2}^{s+1}=\gamma\bar{\gamma}p_{2}=(\prod_{\sigma}(\alpha\bar{\alpha}))p_{2}=2^{ns}p_{2}.$$ By $\delta=\frac{1}{2}(x+y\sqrt{-p_{1}p_{2}})$, $x,y\in \mathbb{Z}$ we get $x^{2}+p_{1}p_{2}y^{2}=2^{ns+2}p_{2}$ which implies $x^{2}\equiv2^{ns+2}p_{2}\pmod {p_{1}}$. Thus $1=\big(\frac{2p_{2}}{p_{1}}\big)=-\big(\frac{p_{2}}{p_{1}}\big)$, and $p_{2}|x$. So that $p_{2}(\frac{x}{p_{2}})^{2}+p_{1}y^{2}=2^{ns+2}$. By the definition of $r$ we get $ns\geq r$.

In summary, there is no GBF with type  $\{2m,n\}$ and $\{m,n\}$ if $2\nmid n$ and $\big(\frac{p_{2}}{p_{1}}\big)=1$, or $2\nmid n<\frac{r}{s}$ and $\big(\frac{p_{2}}{p_{1}}\big)=-1$. This completes the proof of Theorem \ref{thm:35}.
\end{proof}

\begin{remark}
We explain the existence of $r$ in the case $p_{1}\equiv3\pmod8$, $p_{2}\equiv 5\pmod8$ and $\big(\frac{p_{2}}{p_{1}}\big)=-1$. By the definition, $r$ is the least odd integer such that $x^{2}+p_{1}p_{2}y^{2}=2^{r+2}p_{2}$ has a solution $(x,y)\in \mathbb{Z}^{2}$. By FACT (6) of Sect. \ref{sec:4}, we have the decomposition in $T=\mathbb{Q}(\sqrt{-p_{1}p_{2}})$, $p_{2}\mathcal{O}_T=Q^{2}$ and $2\mathcal{O}_T=P\bar{P}$ since $-p_{1}p_{2}\equiv 1\pmod8$. We know that $Q=(p_{2}, \sqrt{-p_{1}p_{2}})$ is not a principal (prime) ideal in $\mathcal{O}_T$ and $[Q]^{2}=[(p_{2}\mathcal{O}_T)]=1$. Thus the order of $[Q]$ is two. We also know that $h(T)=2k, 2\nmid k$ (\cite{Pizer}, Proposition 3 and 4). Let $d$ be the order of $[P]$, then $x^{2}+p_{1}p_{2}y^{2}=2^{d+2}$ has a solution $(x,y)\in \mathbb{Z}^{2}$. Thus $x^{2}\equiv2^{d+2}\pmod {p_{1}}$, so that $1=\big(\frac{2}{p_{1}}\big)^{d+2}=(-1)^{d+2}$ which implies that $2\mid d$ and $d=2e$. Since $d\mid h(T)=2k$ we get $e\mid k$. Therefore $e$ is odd. Both $[P^{e}]$ and $[Q]$ are elements of order two, we get $[P^{e}]=[Q]$ since the size of the class group $C(T)$ is $h(T)\equiv 2\pmod4$ and then it has unique element of order two. Therefore $[P^{e}Q]=1$. Namely, $P^{e}Q$ is a principal ideal $\alpha\mathcal{O}_T$ where $\alpha=\frac{1}{2}(x+y\sqrt{-p_{1}p_{2}}), (\alpha\bar{\alpha})\mathcal{O}_T=(P\bar{P})^{e}Q\bar{Q}=(2^{e}p_{2})\mathcal{O}_T$. Therefore $x^{2}+p_{1}p_{2}y^{2}=2^{e+2}p_{2}, 2\nmid e$. From the minimal property of $r$ we get that $r$ is just $e$, the half of the order of $[P]$.
\end{remark}

\begin{example}
Let $m=p_{1}^{a_{1}}p_{2}^{a_{2}}$ $(a_{1}, a_{2}\geq1)$ where $p_{1}=19\equiv3\pmod8$, $p_{2}=29 \equiv 5\pmod8$, then $\big(\frac{p_{2}}{p_{1}}\big)=\big(\frac{29}{19}\big)=-1$, $f_{1}=18\cdot 19^{a_{1}-1}, g_{1}=1$, $f_{2}=28\cdot 29^{a_{2}-1}, g_{2}=1$ and so $g=g_{1}g_{2}\cdot GCD\{f_{1},f_{2}\}=2, s=\frac{g}{2}=1$. The equation $19x^{2}+29y^{2}=2^{l+2}$ has no solution $(x,y)\in \mathbb{Z}^{2}$ for $l=1,3,5,7,9$ and $11$, but $2^{13+2}=32768=19\cdot 21^{2}+29\cdot 29^{2}$. We get $r=13$ and , by Theorem \ref{thm:35}, there is no GBF with type $\{2\cdot 19^{a_{1}}29^{a_{2}},n\}$ and $\{19^{a_{1}}29^{a_{2}},n\}$ for any $a_{1}, a_{2}\geq1$ and $n=1,3,5,7,9,11$.

\end{example}

\section{Conclusion}\label{sec:6}
In this paper we deal with the existence problem on a generalized bent function $f:Z_{2}^{n} \rightarrow Z_{m}$ , called a GBF with type $\{m,n\}$. In Sect. \ref{sec:2} we proved that there is a GBF with type $\{m,n\}$ if $2\mid GCD\{m,n\}$ or $4\mid m$. For remain case we presented a series of nonexistence results in Sect.\ref{sec:3}-\ref{sec:5} including the following results.

There is no GBF with type $\{m,n\}$ for any odd integer $n\geq1$ if one of the following conditions on $m$ is satisfied ($p, p_{1},\ldots, p_{s}$ denote odd prime numbers, $a, a_{1}, \ldots, a_{s}\geq 1$).

(A). $m=p^{a}$ (Corollary \ref{maincor});

(B). $m=2p^{a}$, $p\equiv3, 5\pmod 8$, or $p\equiv1\pmod 8$ and the order of 2 in $Z_{p}^{*}$ is even (Corollary \ref{cor:semi});

(C). $m=2p_{1}^{a_{1}}p_{2}^{a_{2}}$, $p_{1}\equiv3 \pmod 8$, $p_{2}\equiv5\pmod 8$ and $\big(\frac{p_{2}}{p_{1}}\big)=1$ (Theorem \ref{thm:35});

(D). $m=2p_{1}^{a_{1}}p_{2}^{a_{2}}\cdots p_{s}^{a_{s}}$, $p_{1},\ldots, p_{s}$ are distinct and $p_{1}\equiv p_{2} \equiv\ldots \equiv p_{s} \equiv3\pmod 8$, or $p_{1}\equiv p_{2}\equiv \ldots \equiv p_{s} \equiv5\pmod 8$ (Corollary \ref{maincor});

In Sect. \ref{sec:4}, the decomposition field $M$ of $2$ in $K=\mathbb{Q}(\zeta_{m})$ is real. For all cases in Sect.\ref{sec:5}, $[M,\mathbb{Q}]=2s, 2\nmid s$ and $M$ contains an imaginary quadratic subfield $T$. The nonexistence results are derived by using field-descent method from $K$ to $M$ (Lemma \ref{KD}), then from $M$ to $T$ by norm mapping, and by using some knowledge on the decomposition law of prime numbers in $\mathcal{O}_T$ and the structure of the class group of $T$. Next step we may consider the case that $[M,\mathbb{Q}]=4s, 2\nmid s$ and $M$ contains an imaginary quartic subfield $T$. In this case we would use specific knowledge on integral basis of $\mathcal{O}_T$, the decomposition law in $\mathcal{O}_T$ and the structure of the class group of $T$. The derived conditions on nonexistence of GBF would be more complicated if the Galois group $G(T/\mathbb{Q})$ is cyclic.


\begin{thebibliography}{}

\bibitem{F} K. Feng, Generalized bent functions and class group of imaginary quadratic fields, Sci. China Ser. A, 44, 562--570 (2001)

\bibitem{IR} K. Ireland and M.I. Rosen, A classical introduction to modern number theory, GTM 84, Springer, New York (1990)

\bibitem{KSW} P.V. Kumar, R.A. Scholtz, and L.R. Welch, Generalized bent functions and their properties, Journal of Combinatorial Theory, Series A, 40(1), 90--107 (1985)

\bibitem{LL} T.Y. Lam and K.H. Leung, On vanishing sums of roots of unity, Journal of Algebra, 224(1), 91--109 (2000)


\bibitem{LTH} N. Li, X. T and T. Helleseth, New classes of generalized boolean bent functions over $Z_{4}$, Information Theory Proceedings (ISIT), 2012 IEEE International Symposium on. IEEE, 841--845 (2012)

\bibitem{Pizer} A. Pizer: On the 2-part of the class number of imaginary quadratic number fields, Journal of Number Theory, 8(2), 184-192 (1976).

\bibitem{R} O.S.  Rothaus, On ¡°bent¡± functions, Journal of Combinatorial Theory, Series A, 20(3), 300--305 (1976)

\bibitem{S} K.-U. Schmidt, Quaternary constant-amplitude codes for multicode CDMA, IEEE Trans. on Inf. Theory, 55(4), 1824--1832 (2009)

\bibitem{KUS} K.-U. Schmidt, $Z_{4}$-valued quadratic forms and quaternary sequence families, IEEE Trans. on Inf. Theory, 55(12), 5803--5810 (2009)

\bibitem{ST} P. Sol{\'e} and N. Tokareva, Connections between quaternary and binary bent functions, http://eprint.iacr.org/2009/544.pdf

\bibitem{SM} P. St{\u{a}}nic{\u{a}} and T. Martinsen, Octal bent generalized boolean functions, http://arxiv.org/abs/1102.4812

\bibitem{SMGS} P. St{\u{a}}nic{\u{a}}, T. Martinsen, S. Gangopadhyay and B.K. Singh, Bent and generalized bent boolean functions, Designs, Codes and Cryptography, 69(1), 77--94 (2013)

\bibitem{T} N.N. Tokareva, Generalizations of bent functions. a survey., Journal of Applied and Industrial Mathematics, 5(1), 110--129 (2011)

\bibitem{W} L.C. Washinton: Introduction to cyclotomic fields, GTM 83, Springer, New York (1997)








\end{thebibliography}
\end{document}